\documentclass[
final
]{dmtcs-episciences}

\usepackage[utf8]{inputenc}
\usepackage{subfigure}

\usepackage{xcolor}
\usepackage{graphicx}
\usepackage{booktabs}
\usepackage{caption}
\usepackage{wrapfig}

\newtheorem{theorem}{Theorem}[section]

\newtheorem{lemma}[theorem]{Lemma}
\newtheorem{conjecture}[theorem]{Conjecture}

\graphicspath{{figures/}}

\definecolor {name} {rgb} {0.5,0.0,0.0}

\newcommand{\myparNS}[1]{\noindent{\bfseries #1}}
\newcommand{\mypar}[1]{\smallskip\myparNS{#1}}

\def\QED{{\boldmath$\square$}}
\def\markatright#1{\leavevmode\unskip\nobreak\quad\hspace*{\fill}{#1}}
\def\qed{\markatright{\QED}}

\long\def\conf#1{%
}

\author{Rodrigo I. Silveira\affiliationmark{1}
  \and Bettina Speckmann\affiliationmark{2}
  \and Kevin Verbeek\affiliationmark{2}}
  
  \title{Non-crossing paths with geographic constraints\thanks{A preliminary version of this work appeared in \emph{Proc. 25th International Symposium on Graph Drawing and Network Visualization (GD 2017)}.}}

\affiliation{
  Dept. de Matem\`atiques, Universitat Polit\`ecnica de Catalunya, Spain\\
  Dept. of Mathematics and Computer Science, TU Eindhoven, The Netherlands}
  

\received{2018-3-1}

\revised{2019-1-29}

\accepted{2019-4-23}

\publicationdetails{21}{2019}{3}{15}{4334}

\begin{document}

  \maketitle

  \begin{abstract}
  A \emph{geographic network} is a graph whose vertices are restricted to lie in a prescribed region in the plane. In this paper we begin to study the following fundamental problem for geographic networks: can a given geographic network be drawn without crossings?
We focus on the seemingly simple setting where each region is a vertical segment, and one wants to connect pairs of segments with a path that lies inside the convex hull of the two segments.
We prove that when paths must be drawn as straight line segments, it is NP-complete to determine if a crossing-free solution exists, even if all vertical segments have unit length.
In contrast, we show that when paths must be monotone curves, the question can be answered in polynomial time.
In the more general case of paths that can have any shape, we show that the problem is polynomial under certain assumptions.
 \end{abstract}


\section {Introduction}

Highway, train, and river networks, airline and VLSI routing maps, information flow over the Internet, and the flow of goods and people between different regions all have one thing in common: they can be effectively visualized as a \emph{geographic network}: a graph, whose embedding is fixed, but not completely. The vertices of a geographic network are restricted to lie in a prescribed region while the edges might or might not be required to follow a particular course. In this paper we begin to study the following fundamental problem for geographic networks: can a given geographic network be drawn without crossings?

Many different formulations of this problem exist, which differ in aspects like the shape of the regions, the type of curve used to draw edges, and the type of graph being drawn.
We study the seemingly simple variant where each region is a vertical segment. We restrict the edges to be drawn as simple curves that lie inside the convex hull of the vertical segments corresponding to the end vertices, to force edges to be more or less straight. 

\begin{figure}[t]
	\centering
		\includegraphics{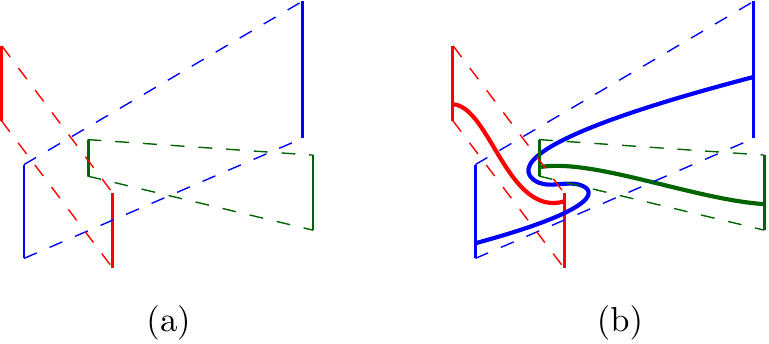}
		\caption{Example of the problem. (a) Instance with three tubes. (b) A possible solution; note that this instance has no solution if all edges must be drawn as straight line segments.}
	\label{fig:example}
\end{figure}

More formally, we are given a graph $G=(V,E)$ and, for each vertex $v \in V$, a vertical segment region $I_v$. For each edge $(u,v) \in E$, we define the \emph{tube} $T_{uv}$ of $(u,v)$ as the convex hull of $I_u \cup I_v$. The goal is to determine if it is possible to draw each vertex $v \in V$ as a point $p_v \in I_v$, and each edge $e=(u,v) \in E$ as a path from $p_u$ to $p_v$ that is contained in $T_{uv}$, such that no two paths cross at a point interior to both.
When this is possible, we say that the tube is \emph{connected}.
Moreover, we assume that $G$ is a matching, so we can solely focus on drawing the edges.
We consider three different ways to draw the edges: straight line segments, $x$-monotone curves, and arbitrary paths.
Figure~\ref{fig:example} illustrates the problem with an example.

\mypar{Related work.}
Drawing graphs or connecting points without crossings has received considerable attention, dating back to classic routing problems (e.g.,~\cite{lp-oprr-83,Lynch75}), and including the vast literature on planar graph drawing (see, e.g.,~\cite{DuncanG13,NishizekiR04,Patrignani13,Vismara13}).

The problems arising when vertices must be inside prescribed (geographic) regions have not been explored that much.
In the limit case of regions that are points---thus fixed vertex positions---there is still freedom on how to draw the edges, leading to a variety of optimization problems, such as drawing edges as polygonal lines with the minimum possible number of bends~\cite{PachW01}, or minimizing the total edge length used~\cite{ChanHKL13}.
The problem changes considerably when there is a continuous set of possible positions for each vertex, as in the problem we focus on in this paper.

To the best of our knowledge, only force-directed layout methods for some particular cases of the problem have been proposed~\cite{aahs-ndgcv-05}.
In that work, each vertex is constrained to lie inside a given polygon, and edges must be drawn as straight line segments. Abellanas et al.~\cite{aahs-ndgcv-05} considered two different goals: minimizing the number of edge crossings and achieving uniform edge length.
However, the methods used were heuristic and of a very different nature than those proposed here.

On the more algorithmic side, there are quite a few problems, studied under many different names, that may resemble the ones studied here, in one or more aspects.
Next we review the most relevant ones.
%
%

Clustered-planarity (c-planarity for short) is a concept developed for drawing clustered graphs.
Essentially, a drawing of a clustered graph is called c-planar if it is a plane drawing, and the vertices of each cluster are drawn inside their own connected region of the plane (see Feng et al.~\cite{FengCE95b}  for a formal definition).
Thus in c-planar drawings, vertices are constrained to lie in certain regions, but these regions are not prescribed, and moreover, the restriction is on the clusters, and not on vertices.
Probably the variant most related to our problem is that of fitting planar graphs to planar maps~\cite{Alam2015}.
In this setting, in addition to the clustered graph, one is given a \emph{compatible} polygonal map, and the goal is to determine if there is a mapping of clusters to polygons that results in a c-planar drawing. Note, however, that the polygons are not pre-assigned to clusters, and, again, that region constraints are for clusters and not for individual vertices.


A setting closer to ours has been considered in the context of data imprecision in computational geometry, where a set of imprecise points is given, meaning that each point can be drawn anywhere inside a given region.
In particular, it has been shown that if the regions are vertical line segments or scaled copies of an arbitrary shape, and the paths are straight line segments, determining if one can draw a cycle without crossings is NP-hard~\cite{loffler11}.
Thus this differs from our setting mainly in the type of graph: cycles, while we focus on matchings.

Aloupis et al.~\cite{ALOUPIS201378} studied the problem of non-crossing matchings between points and geometric objects. In that problem, each edge connects a point to a  geometric object or to a set of points.
Aloupis et al.~\cite{ALOUPIS201378} showed that the problem is polynomially solvable in some special cases, most notably when matching a point to one of two other points, and NP-hard when the number of options increases. In particular, they  showed that our problem for vertical segment regions and straight line segment paths is NP-hard, a fact that was also proved earlier in the Master's thesis of one of the authors of the current paper~\cite{v-ncpfe-08}.
The same problem with unit-size square regions, but drawing general planar graphs instead of matchings, was also shown to be NP-hard~\cite{Angelini2014}.

A problem similar to the second variant studied in this paper, where edges are drawn as $x$-monotone paths, is that of Manhattan-geodesic embeddings of  planar graphs~\cite{Katz2009}.
The goal there is to connect points with non-crossing paths that are rectilinear and $xy$-monotone. This problem was recently shown to be NP-hard~\cite{Klemz2017}. This is in contrast to our problem, which we show to be solvable in polynomial time. Moreover, our problem differs from the Manhattan-geodesic embeddings problem in two important aspects:  endpoints of edges are not fixed, and paths are restricted to lie inside tubes.

Finally, the most relevant previous work in our context deals with the \emph{non-crossing connector problem}, studied by Kratochv{\'{\i}}l and Ueckerdt~\cite{ku-nccp-13}:
given $m$ sets of points $P_i$, $1 \leq i \leq m$, and a region $R_i$ (with $P_i \subset R_i$) for each $i$, the goal is to compute a curve inside each region $R_i$ that goes through all the points in $P_i$ and no two curves cross.
Kratochv{\'{\i}}l and Ueckerdt~\cite{ku-nccp-13} showed that non-crossing connectors always exist if the regions are pseudo-disks  (i.e., the boundaries of any two
regions intersect in at most two points). Their proof proceeds constructively, providing a method to connect the points within each region while avoiding crossings with other connections. For regions that are not pseudo-disks, existence can be decided in polynomial time only for a few cases, while in general the problem is NP-complete.
An important difference with our setting is that \emph{all} given points $P_i$ in each region must be connected.\footnote{Confusingly, Kratochv{\'{\i}}l and Ueckerdt~\cite{ku-nccp-13} quoted one of the current authors as stating that our problem is NP-complete for monotone paths, which is not the case. Also, they claimed that in the aforementioned Master's thesis~\cite{v-ncpfe-08} it was proven that the problem is NP-complete for unit segments, but that is not the case either: the reduction there requires segments of several lengths.}
Still, we will build on top of this result for the case of arbitrary paths, presented in Section~\ref{sec:general}.

\paragraph{Results and organization.}
We study the problem for different restrictions on the path representing the edges. In Section~\ref{sec:hardness} we show that the problem is NP-complete if the paths must be straight line segments, even when the vertical segments defining the tubes have unit length.
In Section~\ref{sec:monotone} we show that, if paths must be $x$-monotone curves, we can decide in polynomial time if a crossing-free drawing exists. For arbitrary paths we can provide such a polynomial-time algorithm only under certain assumptions, as shown in Section~\ref{sec:general}.

%

\section{Straight line paths}
\label{sec:hardness}

In this section we show that, if the edges must be drawn as straight line segments, the problem is NP-complete, even if all the segments defining the tubes have unit length. Let $V = \{v_1, \ldots, v_{2n}\}$ and let $I_i$ be the unit length vertical segment associated with $v_i$. For convenience we assume that there is an edge between $v_{2i-1}$ and $v_{2i}$ ($1 \leq i \leq n$), and let $T_i$ be the corresponding tube.

\begin{wrapfigure}[8]{r}{0.28\textwidth}
  \centering
  \includegraphics[width=0.25\textwidth]{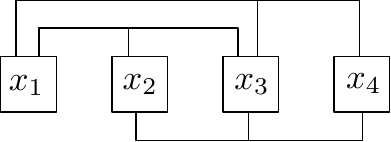}
  \caption{A rectilinear planar embedding of a 3-SAT formula.}
  \label{fig:Rectilinear3SAT}
\end{wrapfigure}

%
%

We prove NP-hardness by reduction from \textsc{Rectilinear planar 3-SAT}~\cite{KnuthR92}.
An instance of this NP-complete problem consists of a 3-SAT formula
and a rectilinear embedding of the graph associated to the formula.
In the embedding all variable vertices lie on a straight line, and clauses are represented as horizontal line segments with at most three vertical line segments that connect to the variables appearing in the clause.
See \figurename~\ref{fig:Rectilinear3SAT} for an illustration of four variables and three clauses. The reduction relies on the following gadgets for variables and clauses.

\begin{figure}[t]
	\centering
		\includegraphics[scale=1,page=1]{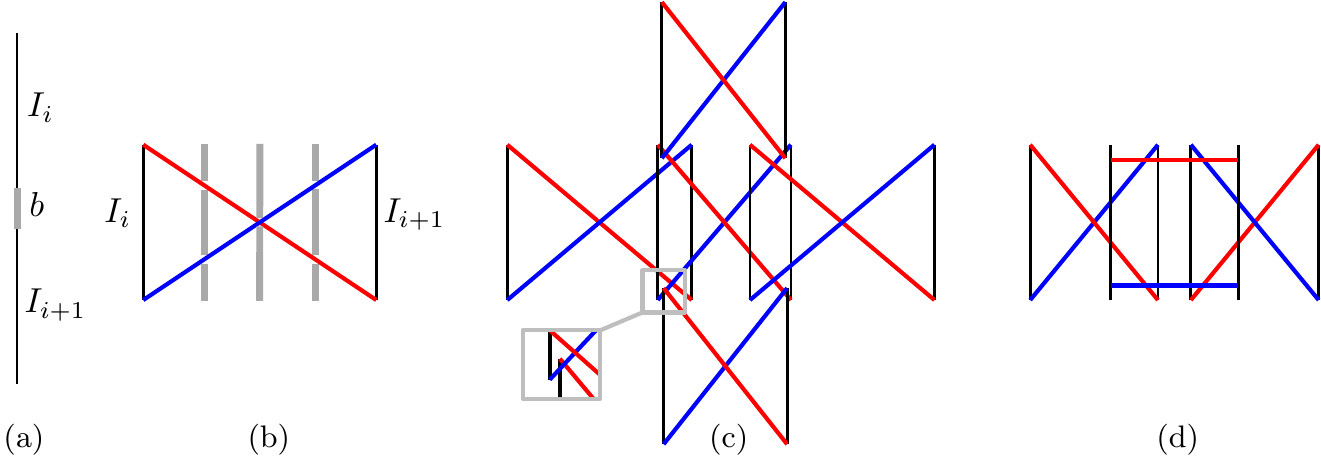}
		\caption{Gadgets in the NP-hardness reduction. Black segments have unit length. (a) A blocker $b$. (b) A basic variable gadget, based on eight blockers (omitting for clarity the unit segments defining each of them). (c) Variable gadgets can be connected to propagate a truth value (blue or red diagonal). Note that blockers are not shown for clarity. (d) Negation of a variable.}
	\label{fig:NP-hardness}
\end{figure}


\mypar{Blockers.}
An essential building block is the construction of vertical edges that cannot be crossed by any segment in a solution, see \figurename~\ref{fig:NP-hardness}(a).
This is achieved by placing a tube connecting two disjoint vertical segments $I_i$ and $I_{i+1}$ exactly above each other,\footnote{Note that the degenerate situation of two equal $x$-coordinates can be avoided by using small perturbations. The same applies to the other gadgets that make use of collinearities: they can all be removed while preserving the behavior of the gadgets.} forcing the segment between $I_i$ and $I_{i+1}$ to be part of any path connecting the tube.

\mypar{Variable gadgets.}
The main component for modeling variables is the \emph{basic gadget} shown in \figurename~\ref{fig:NP-hardness}(b).
Using a small set of blockers, we can limit the possible connections for a tube to only two, shown in blue and red in the figure. These two solutions will correspond to the truth values \emph{true} or \emph{false} of the variable or literal. In general, if we want to limit the possible connections for a tube to a constant number of options, we can enforce this using a constant number of blockers. One generic way to achieve this is to choose three vertical segments (at arbitrary $x$-coordinates) spanning the tube and let these be interrupted by the chosen possible connections. In a non-degenerate situation this will leave only the chosen possible connections as options.
%
As shown in \figurename~\ref{fig:NP-hardness}(c), several basic gadgets can be connected in order to propagate the value in any of four directions.
The value of a variable can be negated by adding a tube with two horizontal segments as options, as shown in \figurename~\ref{fig:NP-hardness}(d).


\begin{figure}[t]
	\centering
		\includegraphics[scale=1,page=1]{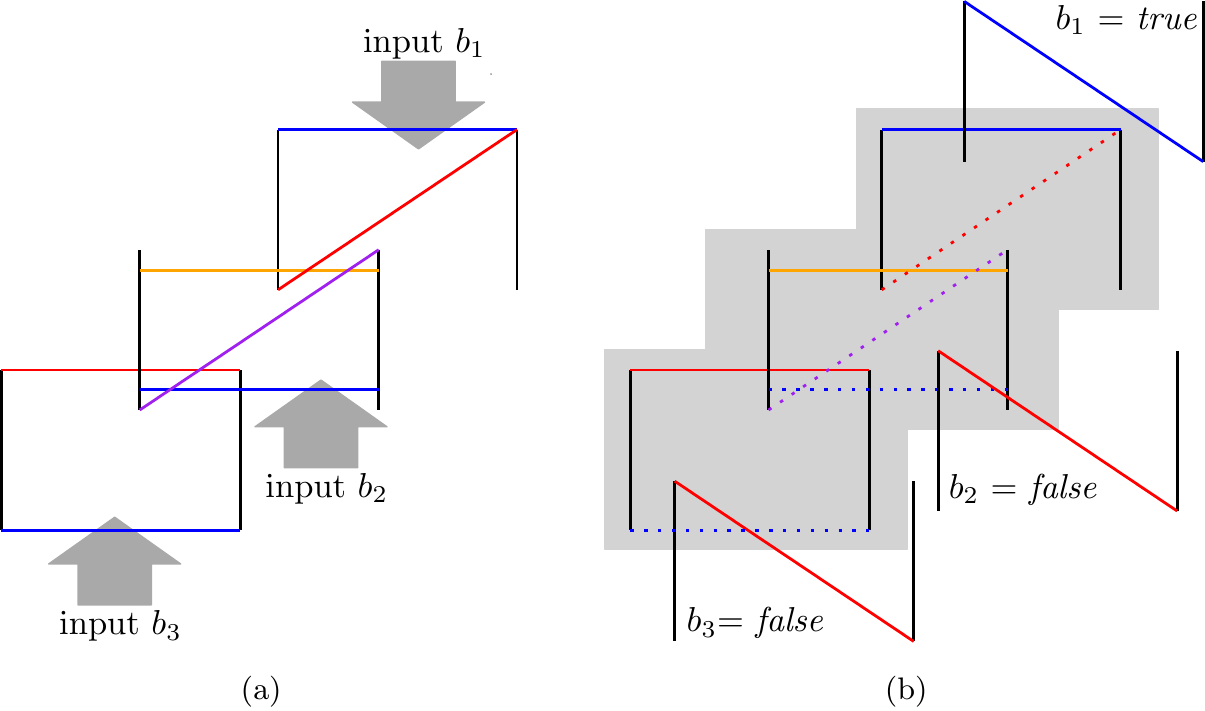}
		\caption{Clause gadgets. (a) The gadget consists of three tubes. The three literals $\{b_1,b_2,b_3 \}$ that participate in the clause cross the blue edges when their value is \emph{false}. (b) Example of the only non-crossing solution when the first literal is \emph{true} and the other two are \emph{false} (dotted edges create crossings).}
	\label{fig:NP-hardness-clause}
\end{figure}

\mypar{Clause gadgets.}
In the embedding given in \textsc{Rectilinear 3-SAT}, a clause is represented by a horizontal line segment with three vertical segments, which connect to the variables.
A horizontal segment can be recreated by using a single tube wide enough.
Vertical segments can be represented by a chain of vertically stacked tubes (see \figurename~\ref{fig:NP-hardness}(c)).
The most interesting part of the clause is the point at which the three paths connect, in which the values of the three literals interact.
In our gadget, this is achieved by using three tubes, as shown in \figurename~\ref{fig:NP-hardness-clause}.
The top and bottom tubes have only two possible paths connecting them
(for clarity, in the figures we omit the blockers needed to force this situation).
The middle tube can be connected with three different edges.
The three literals that form the clause attach to it through the blue edges.
More precisely, a literal will have an edge crossing with one of the blue edges of the clause if and only if its value is \emph{false}.
The key property of the clause gadget is that there exist non-crossing paths connecting the three tubes if and only if at least one literal is \emph{true}.
Note that the variable gadgets do not all connect to the clause gadgets from the bottom. However, we can easily achieve this construction by minor modifications to the rectilinear embedding and using the construction in \figurename~\ref{fig:NP-hardness}(c).


With this construction we obtain the desired NP-hardness reduction.
The following lemma is a direct consequence of the construction.


\begin{lemma}
A satisfying truth assignment for the variables in the 3-SAT formula exists if and only if all tubes in the associated construction can be connected with straight line segments without crossings.
\end{lemma}
\begin{proof}
First consider an input 3-SAT formula with a satisfying truth assignment.
Each variable value is mapped to one of the two possible edges (red or blue) in its corresponding variable gadget.
By construction, all the tubes that propagate this value can be connected without crossings as well, choosing the same edge color (or the opposite, when they need to propagate the value to a clause where the variable appears negated).
Since each clause can be satisfied with the chosen truth assignment, at least one of the input values to the clause must be true.
Then the key property of the clause gadget implies that there is a way to connect the three tubes forming the clause without crossings. 
Therefore all tubes can be connected without crossings.
Conversely, if all tubes in the construction can be connected without crossings, each variable is associated with a single truth value consistently. Moreover,  the tubes of each clause gadget can also be connected without crossings. Then the key property of the clause gadget implies that each clause can be satisfied, and overall that a satisfying truth assignment exists.
\end{proof}

It remains to show that the problem is in NP.


\begin{lemma}
Given $n$ tubes defined by unit vertical segments, the problem of deciding if the tubes can be connected with straight line segments is in NP.
\end{lemma}
\begin{proof}
Note that it is very easy to check whether there
is a crossing or not, given the coordinates of the segment of each tube. But we also need to show that these coordinates can be represented by a polynomial number of
bits. In order to show this, we look at the problem differently.

A solution to the problem can be described as two $y$-coordinates for each tube, $y_i, y_{i+1}$, representing the $y$-coordinates of the endpoints of the segment.
Then the problem can be seen as finding some constraints on the $y$-coordinates $y_i$. First consider the constraint that the line segments must be non-crossing.
It is easy to see that this can be expressed as a series of \emph{orientation tests}, in which one checks if a point is to the left or right of an oriented line through two other points.
An orientation test consists of checking the sign of the determinant of a $3 \times 3$ matrix, which in this case has only one row of variables (the $y$-coordinates of the three points in question).
Therefore this test results in a linear expression.

So checking if two line
segments are non-crossing is a boolean formula on linear constraints. The constraint that an endpoint of a segment is on a vertical segment is clearly also linear.

As described above, the problem consists in finding a vector $(y_1, y_2, \ldots, y_{2 n})$ satisfying a boolean formula on linear constraints. Assume that the problem has a solution
$(y_1, y_2, \ldots, y_{2 n})$. This solution must be in a $2 n$-dimensional cell bounded by the hyperplanes representing the linear constraints (or it is on an intersection of
hyperplanes). Such a cell is bounded, due to the constraints forcing the endpoints of the segments to be on the given vertical segments.
 We can simply choose a solution that is on one of the corner points of the cell. Note that this
solution is on an intersection of hyperplanes. Because the intersection of a collection of $2 n$ hyperplanes can be represented by a polynomial number of bits, the solution of a
problem instance can be as well. Finally, given the vector $(y_1, y_2, \ldots, y_{2 n})$, the linear constraints can easily be checked. Thus, the problem is in NP.
\end{proof}

\begin{theorem}
Given $n$ tubes defined by unit vertical segments, deciding if the tubes can be connected with straight line segments is NP-complete.
\end{theorem}

\section{Monotone paths}
\label{sec:monotone}

In this section we consider edges drawn as $x$-monotone paths. 
A path is said to be $x$-monotone if no vertical line intersects the path more than once, see \figurename~\ref{fig:example_monotone} for an illustration.

\begin{figure}[t]
	\centering
		\includegraphics{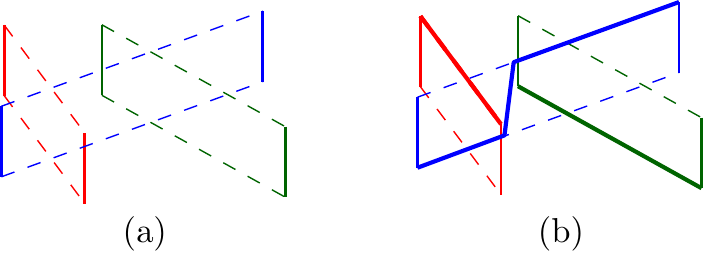}
		\caption{Example using monotone paths. (a) Instance with three tubes. (b) A solution using $x$-monotone paths; note that this instance has no solution if all edges must be drawn as straight line segments.}
	\label{fig:example_monotone}
\end{figure}

\begin{wrapfigure}[10]{r}{0.14\textwidth}
  \begin{center}
    \includegraphics[width=0.12\textwidth]{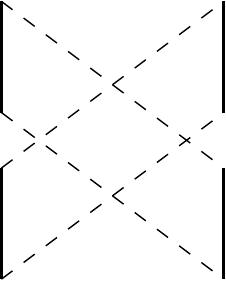}
  \end{center}
  \caption{A full crossing.} \label{fig:FullCrossing}
\end{wrapfigure}
We start with some observations that hold for arbitrary paths.
%
We say that two tubes \emph{fully cross} if the vertical segments are completely disjoint from the other tube, and the 
intersection of the two tubes is nonempty (see \figurename~\ref{fig:FullCrossing}).
The first basic observation is that whenever two tubes fully cross, no solution can exist.
Therefore we assume from now on that no two tubes fully cross.
The most interesting cases occur when two tubes intersect, without fully crossing.
This necessarily happens because (at least) one of the vertical segments of a tube intersects the other tube.
Figure~\ref{fig:Intersections} shows examples of such situations.
We distinguish between \emph{single intersections}, where only one tube segment intersects another tube, or \emph{double intersections}, where two different segments intersect another tube (either both from the same tube, or one from each).

\begin{figure}[t]
	\centering
		\includegraphics[scale=1,page=1]{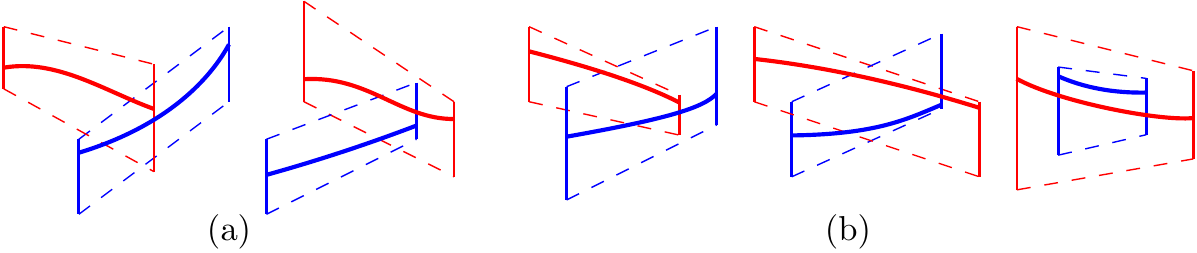}
		\caption{Examples of tube intersections and solutions: (a) single and (b) double. Double intersections also admit solutions with the inverse vertical red/blue order.}
	\label{fig:Intersections}
\end{figure}

Single intersections locally induce a vertical order between the paths in any solution.
For instance, in the situations in \figurename~\ref{fig:Intersections}(a), the red tube can be considered \emph{above} the blue one, because in any solution the red path will be above the blue one at the $x$-coordinate equal to that of the vertical segment creating the intersection.
On the other hand, no such order exists for a double intersection.
Indeed, in any double intersection there are solutions with both orders of the paths in the tubes, see \figurename~\ref{fig:Intersections}(b).
Based on this we define the \emph{order graph}.

\mypar{Order graph.} The order graph of a set of tubes has a vertex for each tube and a directed edge from $T_i$ to $T_j$ if $T_i$ and $T_j$ have a single intersection where $T_j$ is above $T_i$.
We also add a directed edge from $T_i$ to $T_j$ if $T_i \cap T_j = \emptyset$ and $T_i$ and $T_j$ share an $x$-coordinate where $T_j$ is above $T_i$.
Double intersections are not represented in the order graph.
The order graph encodes enough information to decide whether a solution exists using $x$-monotone paths.


\begin{theorem}
\label{thm:monotone}
Given a set of tubes defined by vertical segments, the tubes can be connected with $x$-monotone paths if and only if the order graph is acyclic and no two tubes fully cross.
\end{theorem}

\begin{proof}
First we prove that if the order graph is acyclic, and no two tubes fully cross, then there exists a solution.
The directed edges in the order graph induce a partial order that can be extended to a total order on the tubes. Let $T_1,\dots,T_n$ be that order from bottom to top.
Let $\ell_i$ denote the bottom side of tube $T_i$.
We maintain the following invariant: for $i=1,\dots,n$, path $p_i$ of tube $T_i$ consists of parts of $\ell_j$ with $1 \leq j \leq i$ and vertical segments. We can clearly draw $p_1$ along $\ell_1$.
Suppose now that we want to draw path $p_i$ ($i > 1$). We start $p_i$ at the highest intersection of the left vertical segment with any path $p_j$ ($j < i$), or at $\ell_i$ if no such path exists. We follow a restricting path $p_j$ until the right side of $T_j$, after which we drop down vertically, hitting either another path $p_k$ or $\ell_i$. In the latter case, or if we already hit $\ell_i$ before reaching the right side of $T_j$ while following $p_j$, we can follow $\ell_i$ until hitting another restricting path. We then repeat this process until we reach the right side of $T_i$. The resulting path $p_i$ only follows paths $p_j$ ($j < i$), vertical segments, and $\ell_i$, and thus satisfies the invariant. Finally note that $p_i$ can leave $T_i$ only if it is restricted by a path $p_j$ intersecting the top of $T_i$. By the invariant, there must be some $\ell_k$ ($k < i$) intersecting the top of $T_i$,  violating the order.

It remains to prove the result in the opposite direction.
It is clear that if two tubes fully cross, there is no solution, so we focus on the case of a directed cycle in the order graph.

First we observe that any solution with $x$-monotone paths can be extended to a set of unbounded monotone curves, as follows.
Take the path corresponding to any tube in the cycle and extend it towards both sides horizontally.
If while doing that the extended path hits another path, follow it without crossing it until it ends, and then continue horizontally.
Now, repeat that for each of the tubes in the solution.
If while extending the path of a tube it hits the extended path of another tube, from then on both paths go together. See \figurename~\ref{fig:ExtendingPaths}.

\begin{figure}[tb]
	\centering
		\includegraphics[scale=1,page=1]{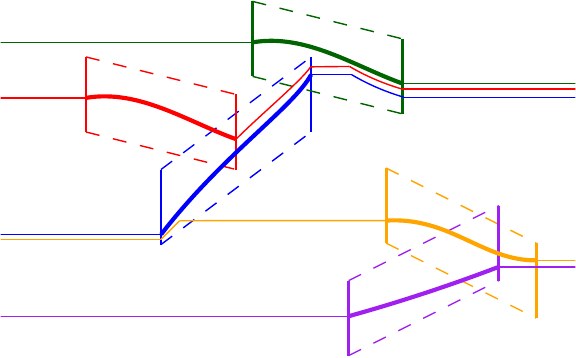}
		\caption{Extending monotone paths to monotone unbounded curves.}
	\label{fig:ExtendingPaths}
\end{figure}

Since the initial paths were $x$-monotone, the extensions are also $x$-monotone, because each of them consists of a concatenation of horizontal segments and parts of $x$-monotone paths.
The final result is a set of $x$-monotone unbounded curves that respect all the order relations between the tubes, namely, if $T_j$ is above $T_i$, then the extended path of $T_j$ is above that of $T_i$. Since the $x$-monotone unbounded curves are totally ordered, and this order respects all order relations in the order graph, the order graph must be acyclic. Thus, if the order graph contains a directed cycle, then no solution with $x$-monotone curves can exist. 
%
%
\end{proof}

\noindent Therefore, to solve the problem it is enough to compute the order graph and check the two conditions in the theorem. This can be done in $O(n^2)$ time, since the order graph can have at most quadratic size.



\section{Arbitrary paths}
\label{sec:general}

\begin{figure}[t]
\centering
\includegraphics{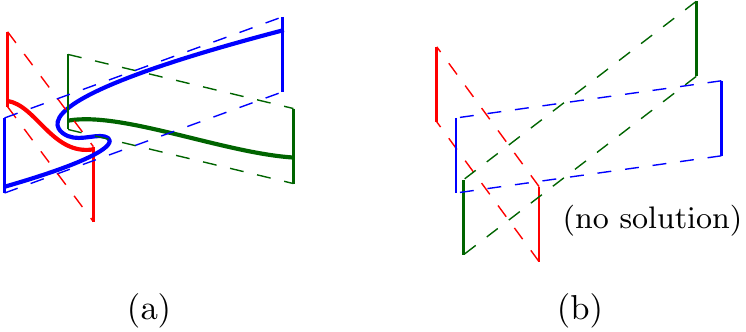}
\caption{Examples of two cycles in the order graph, one that has a solution (a) with arbitrary paths, and one that does not (b).}
\label{fig:CycleWithSolution}
\end{figure}

If we allow edges to be drawn by arbitrary paths, then a cycle in the order graph can sometimes be realized, as shown in \figurename~\ref{fig:CycleWithSolution}(a)---the cycle here is blue $\rightarrow$ red $\rightarrow$ green $\rightarrow$ blue. However, that is not always the case, as the example in  \figurename~\ref{fig:CycleWithSolution}(b) shows.


Nevertheless, if we disallow double intersections, then we can still decide in polynomial time whether a solution exists.
The key idea is to use a result by Kratochv{\'{\i}}l and Ueckerdt~\cite{ku-nccp-13} that states that if the regions (in our case, tubes) form a set of pseudo-disks, then there is always a solution.
Two tubes with a single intersection may not be pseudo-disks, but we can try to convert them into pseudo-disks by cutting off parts that cannot be used in any solution.
This leads to a procedure that allows us to determine in polynomial time if a solution exists.

\begin{theorem}
\label{thm:pseudodisks}
Given a set of tubes defined by vertical segments such that no two tubes form a double intersection, one can determine if all the tubes can be connected without crossings using arbitrary paths in polynomial time.
\end{theorem}

\begin{proof}
Kratochv{\'{\i}}l and Ueckerdt~\cite{ku-nccp-13} showed that the non-crossing connectors problem always has a solution when the regions form a collection of pseudo-disks~\cite[Theorem 2]{ku-nccp-13} (i.e., the boundaries of any two
regions intersect in at most two points).
In our context, the regions are the tubes.
To apply their result to our problem we need two things.

First, the tubes need to be pseudo-disks.
If no two tubes fully cross or create a double intersection, the only way in which they can interact is through single intersections.
Two tubes that intersect in a single intersection are not always pseudo-disks, since the tube boundaries can intersect in four points.
However, it is possible to make them pseudo-disks by cutting off the part of one of the tubes that sticks out of the other, as shown in \figurename~\ref{fig:Pseudodisks}(a).
We refer to this part as an \emph{ear}.
Initially, an ear is a triangle with one vertical edge, which is part of the vertical segment of a tube, and two other edges that are parts of the boundaries of the two tubes involved.
The path order forced in every single intersection implies that in any solution, the ear will be separated from its tube by the path of the other tube, meaning that there can be no path inside the ear. Thus we can cut off the ear without affecting any solution.
For instance, in \figurename~\ref{fig:Pseudodisks}(a), the blue path can never enter the ear without crossing the red path.

\begin{figure}[tb]
  \center
  \includegraphics{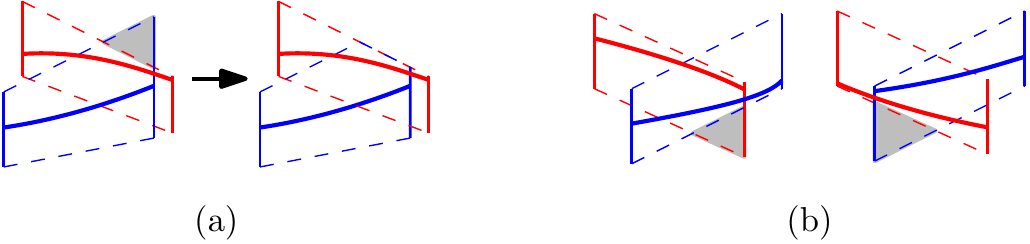}
  \caption{(a) Cutting off an ear (shaded gray) so that two tubes that have a single intersection become pseudo-disks; the gray region is guaranteed to be empty in any solution, if one exists.
  (b) The approach does not extend to double intersections because then one does not know in advance which of the two gray regions will be empty.
  }
  \label{fig:Pseudodisks}
\end{figure}

Cutting off an ear can introduce a full crossing between two tubes. Figure~\ref{fig:CuttingCorners} shows an example where the removal of an ear that contains part of the vertical segment of another tube generates a full crossing. Since we only cut off parts of the tube that cannot be used, we can conclude in this situation that the original problem has no solution. On the other hand, cutting off an ear cannot introduce a double intersection: tubes only become smaller, and thus, a vertical segment that is disjoint from a tube will remain disjoint from that tube.

We now use the following simple approach. While there exists a pair of tubes with more than two intersections (that is, the current tubes are not pseudo-disks), we cut off the corresponding ear. This ear always exists, as we have only single intersections (although it needs not to be triangular). After this procedure ends, the set of tubes are pseudo-disks.

We need to argue that this procedure terminates after a polynomial number of steps. To that end, note that the complexity of the arrangement formed by the tubes cannot increase. Furthermore, an ear always consists of a non-empty set of faces of the arrangement of tubes. In particular, in every step there is a tube that loses at least one face of the arrangement. Since the complexity of the arrangement is $O(n^2)$, the number of steps in this procedure is $O(n^2)$ as well.

The result of the procedure is a set of truncated tubes with the same solutions as the original ones.
To determine if a solution exists, it is enough to check if any two of the truncated tubes fully cross.
If two truncated tubes cross, there is no solution for the original tubes either.
If no two fully cross, then the truncated tubes are pseudo-disks.

Second, to apply the result of Kratochv{\'{\i}}l and Ueckerdt~\cite[Theorem 2]{ku-nccp-13}, we need a discrete set of points to be connected inside each tube.
It is enough to place an endpoint on each vertical segment of the truncated tubes.
We obtain a set of regions that are pseudo-disks with two endpoints in each. 
The aforementioned result guarantees that they have a solution, and thus the original set of tubes also does.
\end{proof}


\noindent We note that the technique used above to cut off ears in single intersections to make them pseudo-disks does not extend to double intersections.
As \figurename~\ref{fig:Pseudodisks}(b) shows, since the order of the two paths in a double intersection is not fixed, it is not possible to know in advance which of the two ears will be empty in a solution.

\begin{figure}[tb]
  \center
  \includegraphics{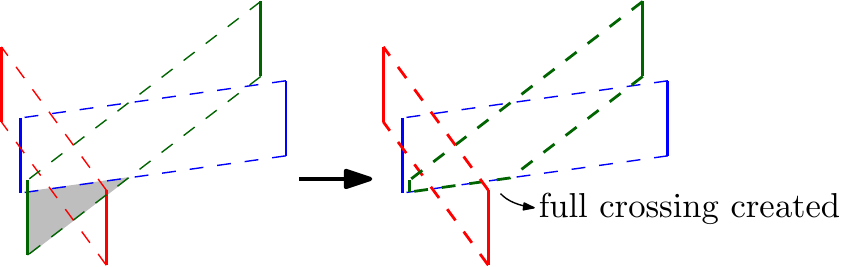}
  \caption{When the ear between the blue and green tubes (shaded gray) is cut off, the single intersection between green and red becomes a full crossing.  This shows that the original problem has no solution.}
  \label{fig:CuttingCorners}
\end{figure}


\section{Conclusions and open problems}
In this work we began the study of a fundamental problem for geographic networks: deciding if a given network can be drawn without crossings.
We have focused on a simple variant where regions are vertical segments, and the graph is a matching.
Despite the apparent simplicity of the problem, it turned out to be very challenging, with varying problem complexities depending on the way in which edges can be drawn.
There are multiple directions in which this work can be continued. 


To begin with, in the restricted setting studied here, the arguably most important variant of arbitrary paths is still not fully solved: our algorithm only works in the absence of double intersections.
If we allow double intersections, the problem remains open. However, we expect the problem to remain polynomial also in this case. Furthermore, we conjecture that the problem admits a Helly-type property, implying a polynomial time algorithm.

\begin{conjecture}
There exists a universal constant $C$ such that, if a set of tubes defined by vertical segments does not admit a solution with arbitrary paths, then there exists a subset of at most $C$ tubes that also does not admit a solution.
\end{conjecture}

A whole set of open questions arises when the setting is made more general. 
A first step would be, for example, to consider tubes defined by line segments with more than one orientation (say, using $k>1$ different slopes).
This would probably require a different approach.
Then it is most interesting to understand what is the simplest type of 2-dimensional region  that can be solved in polynomial time when the graph is matching, as in this work.
Going in a different direction, it would also be worthwhile to consider more general families of graphs, such as paths or trees, even if it is for tubes defined by vertical segments.




%

\vspace{10pt}
\mypar{Acknowledgements.} 
We are grateful to the anonymous reviewers for their detailed and useful comments.
R.I.S was partially supported by projects MTM2015-63791-R (MINECO/ FEDER) and Gen.\ Cat.\ 2017SGR1640, and by MINECO's Ram{\'o}n y Cajal program. B.S. and K.V. are supported by the Netherlands Organisation for Scientific Research (NWO) under project no.~639.023.208 and 639.021.541, respectively.

\bibliographystyle{abbrv}
\bibliography{refslong}

\end{document}